\documentclass[conference]{IEEEtran}

\usepackage{amssymb,amsmath,amsfonts}
\usepackage[]{algorithm2e}
\usepackage{blkarray}
\usepackage{bbm}
\usepackage{bm}
\usepackage{cases}
\usepackage{cite}
\usepackage[usenames,dvipsnames]{color}
\usepackage{dsfont}
\usepackage{float}
\usepackage{graphicx}
\usepackage{hyperref}
\usepackage[latin9]{inputenc}
\usepackage{mathtools}
\usepackage{mathrsfs}
\usepackage{setspace}
\usepackage{soul}
\usepackage{subfigure}
\usepackage{stfloats}
\usepackage{tikz}
\usepackage{url}
\usepackage{verbatim}
\usepackage{xspace}
\usepackage{algpseudocode}

\usepackage{caption,epstopdf}

\newtheorem{lem}{Lemma}
\newtheorem{ass}{Assumption}
\newtheorem{theorem}{Theorem}

\def\mb{\mathbf}

\def\mc{\mathcal}

\IEEEoverridecommandlockouts
\makeatletter
 \let\old@ps@headings\ps@headings
 \let\old@ps@IEEEtitlepagestyle\ps@IEEEtitlepagestyle
 \def\confheader#1{%
 \def\ps@headings{%
 \old@ps@headings%
 \def\@oddhead{\strut\hfill#1\hfill\strut}%
 \def\@evenhead{\strut\hfill#1\hfill\strut}%
 }%
 \def\ps@IEEEtitlepagestyle{%
 \old@ps@IEEEtitlepagestyle%
 \def\@oddhead{\strut\hfill#1\hfill\strut}%
 \def\@evenhead{\strut\hfill#1\hfill\strut}%
 }%
\ps@headings%
 }
\makeatother
\confheader{
}
\usepackage[pscoord]{eso-pic}

\begin{document}
\title{\bf
Robust-to-Noise Algorithms for Distributed Resource Allocation and Scheduling }

\author{\IEEEauthorblockN{ Mohammadreza Doostmohammadian}
\IEEEauthorblockA{
\textit{Faculty of Mechanical Engineering}, \\
\textit{Semnan University}, Iran, \\ \texttt{doost@semnan.ac.ir}}
\and
\IEEEauthorblockN{Alireza Aghasi}
\IEEEauthorblockA{\textit{Electrical Engineering and Computer Science Department}, \\ \textit{ Oregon State University},
USA, \\
\texttt{alireza.aghasi@oregonstate.edu}}
}

\maketitle

\begin{abstract}
	Efficient resource allocation and scheduling algorithms are essential for various distributed applications, ranging from wireless networks and cloud computing platforms to autonomous multi-agent systems and swarm robotic networks. However, real-world environments are often plagued by uncertainties and noise, leading to sub-optimal performance and increased vulnerability of traditional algorithms. This paper addresses the challenge of robust resource allocation and scheduling in the presence of noise and disturbances.
	The proposed study introduces a novel sign-based dynamics for developing robust-to-noise algorithms distributed over a multi-agent network that can adaptively handle external disturbances. Leveraging concepts from convex optimization theory, control theory, and network science the framework establishes a principled approach to design algorithms that can maintain key properties such as resource-demand balance and constraint feasibility. Meanwhile, notions of uniform-connectivity and versatile networking conditions are also addressed.
	
\keywords  distributed optimization, graph theory, distributed resource allocation, random noise, optimality gap
\end{abstract}

\section{Introduction} \label{sec_intro}
The advent of cloud-computing, Internet-of-Things (IoT), and parallel processing has motivated distributed algorithms in many disciplines including distributed estimation and filtering \cite{acc13,das2015distributed,deghat2019detection,icassp13}, distributed target tracking \cite{ennasr2016distributed,tase}, distributed robot coordination and coverage control \cite{cortes2004coverage,msc09,sayyaadi2010distributed}, and distributed machine learning \cite{csl2021,xin2020decentralized}. In such distributed setups the information shared over the multi-agent network might be subject to noise and disturbances. Many works in the literature are devoted to designing resilient and robust algorithms to compensate for the effect of noise and ensure convergence in noisy environments. One of the applications of such resilient distributed algorithms is in distributed scheduling and allocation of resources over a multi-agent network. Distributed resource allocation (DRA) refers to the problem of optimally assigning (and distributing) certain resources among a group of entities while (i) meeting the output demand and preserving the resource-demand balance (referred to as constraint feasibility), and (ii) minimising the utility costs modelled via some objective functions. In this paper, we are interested in designing DRA algorithms robust to noise and disturbances. This benefits the applications in noisy setups and environments and the results can be further extended to other distributed setups such as distributed estimation and fault detection.

\subsection{Literature Review and Contributions}
Sign-based algorithms are known to be inherently resilient to impulsive noise and are used in many setups, including diffusion strategy for distributed learning
\cite{zayyani2016distributed}, consensus algorithms \cite{stankovic2020nonlinear,taes,stankovic2019robust}, and distributed estimation \cite{sadigh2022proportionate,babarsad2019analytic,zandi2022diffusion} among others. This motivates this paper to adopt sign-based dynamics to solve the DRA problem. Similarly in the literature, some works are devoted to designing robust algorithms in similar resource allocation setups. For example, robust economic dispatch problem with quadratic energy cost objectives is considered in \cite{baranwal2020robust,kar2012distributed}. Network resource allocation under uncertainty is solved in \cite{doan2020distributed} under decaying step-sizes and all-time connected networks. The work \cite{ferguson2023robust} designs game-theoretic solutions robust to failures and stubborn agents. Robust and secure solutions in wireless sensor network setups with applications in communication networks are considered in \cite{xu2022robust}. In \cite{ding2021differentially} the authors propose deviation tracking strategies to tackle faulty information from adversary agents and secure privacy. Attack-resilient DRA solutions are also proposed in \cite{turan2020resilient}. A perturbation-based method to solve DRA without closed-form expression of the cost function is proposed in \cite{ramirez2016gradient}. Asynchronous consensus-based protocols for DRA over time-varying all-time connected networks are discussed in \cite{alaviani2021reciprocity}. General
nonlinear models for DRA that span a wide range of applications are proposed in \cite{ecc22,ccta22}. Our paper proposes robust-to-noise algorithms which are all-time feasible up-to certain point (preserve the balance between the resource-demand) over uniformly-connected versatile networks that may lose connectivity at some times. The solution is not constrained to quadratic costs but may be applied to general non-quadratic costs (even subject to penalty terms). The paper advances the state-of-the-art by introducing sign-based dynamics to make the solution robust to noise. Further, we improve the convergence rate by modifying the solution dynamics.

\subsection{Applications}
\textbf{Generator Coordination and Energy Management:}
Distributed energy resource management is a control mechanism to keep the balance between the generated power and load demand over the smart grid. This process determines the optimal schedule for operating power generation units within a power system over a specific time horizon. The idea is to locally assign the power rates to a group of generators and coordinate them to maintain the load demand while minimizing the power generation costs
\cite{cherukuri2015distributed,ccta22,cherukuri2017distributed,LCSS23}. This problem is sometimes referred to as the economic dispatch problem.

\textbf{CPU Scheduling:}
CPU scheduling refers to the process of determining which process or task should be allocated to the CPU (Central Processing Unit) at a given point in time. In modern operating systems, multiple processes or tasks may be competing for CPU time simultaneously. CPU scheduling algorithms are used to decide which task to be assigned to which server or computing node. In other words,
CPU scheduling  efficiently allocates the computational tasks across multiple CPUs or servers while minimizing certain processing cost function
\cite{grammenos2023cpu,OJCSYS,rikos2021optimal}.

\textbf{Scheduling Plug-in Electric Vehicle (PEV) charging:}
A group of PEVs needs to select a schedule to charge their internal batteries while minimizing the aggregate electricity costs at all charging stations. The formulation is subject to both vehicle-level constraints
(e.g., physical limitations of the batteries) and grid-wide constraints (maximum deliverable power). The coupling constraint among the group of nodes here is that the provided charging power is equal to the battery loads. The objective function to be minimized represents the sum of overall charging costs at the stations. Both centralized and distributed solutions for this problem are provided in the literature
\cite{falsone2020tracking,mukherjee2016distributed,vujanic2016decomposition,falsone2017dual}.

\subsection{Paper Organization}
Section \ref{sec_fram} establishes the mathematical framework for the problem of resource allocation and scheduling. Section \ref{sec_solution} provides the main proposed dynamics to solve the problem along with the convergence analysis. Section~\ref{sec_sim} presents the simulation results to verify theoretical concepts, and finally, Section~\ref{sec_con} concludes the paper.

\section{The Framework} \label{sec_fram}
The DRA problem is mathematically formulated as follows:
\begin{equation} \label{eq_dra}
	\begin{aligned}
		\displaystyle
		\min_\mb{x}
		~~ & F(\mb{x}) = \sum_{i=1}^{n} f_i(\mb{x}_i) \\
		\text{s.t.} ~~&  \mb{a}^\top \mb{x} = b
	\end{aligned}
\end{equation}
with $x_i \in \mathbb{R}$ as the state (resource) of node $i$,
$\mb{x} = [x_1;\dots;x_n] \in \mathbb{R}^{n}$ as the global state variable (the resource vector), $\mb{a}=[{a}_1;\dots;{a}_n]^\top \in \mathbb{R}^n$ as the state weighting factor and $b   \in \mathbb{R}$ as the demand/load which needs to balance with the sum of weighted resources. All the local cost functions $f_i(x_i):\mathbb{R} \mapsto \mathbb{R}$ are assumed to be strictly convex.

\begin{lem}\label{lem_optimal_solution}
	The problem \eqref{eq_dra} has a unique optimizer $\mb{x}^*$ satisfying the following property:
	\begin{eqnarray}
		\nabla F(\mb{x}^*) = \bm{\varphi}^* \mb{a},
	\end{eqnarray}
	with $\bm{\varphi}^* \in \mathbb{R}$ and $\nabla F(\mb{x}^*) = [\nabla f_1(x^*_1);\dots;\nabla \widetilde{f}_n(x^*_n)]$ as the global gradient at $\mb{x}^*$.
\end{lem}
\begin{proof}
	The proof for $a_i =1$ follows the KKT condition as described in \cite{boyd2006optimal}. Adopting a similar methodology one can generalize the proof for general vector $\mb{a}$.
\end{proof}

The distributed resource allocation solution is defined over a communication network (or data-sharing network). In this paper, we assume a uniformly connected network that is connected over every finite time-interval $B$ (also referred to as the B-connectivity condition). This is motivated by the fact that mobile robotic networks and multi-agent systems may lose connectivity at some time-instants and regain connectivity at some other times. This is in contrast to many existing literature that requires all-time network connectivity for convergence, e.g., see \cite{cherukuri2015distributed,boyd2006optimal,shames2011accelerated,turan2020resilient}. In this context, this work advances state-of-the-art by relaxing the network connectivity condition for distributed resource allocation.

In the case of box constraints in the form $x_{min}<x_i<x_{max}$ one can use extra additive penalty or barrier functions in the objective to address these local constraints. Different penalizing functions for box constraints are proposed in the literature, e.g., see \cite{nesterov1998introductory,bertsekas1975necessary,smith1997penalty}.

\textit{Our contribution} is to design distributed algorithms to solve the above DRA problem locally with robust algorithms that can compensate for the effect of noise and disturbances over the communication network. Our proposed solution converges to the optimal resource allocation with a certain optimality gap in the presence of  noise. This optimality gap is due to the effect of noise.
The solution is proposed based on the signum function and further we improve its convergence rate with some additive terms. The solution is described in the next section.

\section{The proposed Sign-based Robust Solution} \label{sec_solution}
Motivated by the nonlinear consensus algorithms, we propose a general nonlinear sign-based dynamics to solve problem \eqref{eq_dra}. The proposed general solution is as follows:
\begin{align}
	\dot{x}_i = -\frac{1}{a_i}\sum_{j \in \mc{N}_i} W_{ij} sgn\Big(\frac{\partial f_i(x_i)}{a_i} -  \frac{\partial f_j(x_j)}{a_j} + \nu_{ij} \Big),
	\label{eq_sol}
\end{align}
with  $W_{ij}$ as the $(i,j)$ link weight, $\nu_{ij}$ as the additive noise over this link, $\mc{N}_i$ as the neighborhood of node $i$, $\partial f_i(x_i)$ as the gradient of cost function $f_i$  with respect to $x_i$, and the sign function defined as
\begin{align} \label{eq_sgn}
	sgn(x) = \frac{x}{|x|}.
\end{align}
Note that sign function $sgn(\cdot): \mathbb{R} \mapsto \mathbb{R}$ is a nonlinear and odd mapping, i.e., $sgn(-x) = - sgn(x)$ and $sgn(x)=1>0$ for $x>0$, $sgn(x)=-1<0$ for $x<0$, and $sgn(x)=0$ for equilibrium $x=0$. In the presence of random additive noise $\nu$ the sign function works as follows:
\begin{align} \label{eq_sgn_noise}
	sgn(x+\nu) = \left\{ \begin{array}{ll}
		sgn(x), & \text{for}~  |x|>|\nu|,\\
		\pm sgn(x), & \text{for}~|x|\leq |\nu|.
	\end{array}\right.
\end{align}
The above implies that, in the presence of random noise, for small values of $|x|$ (near the equilibrium) the noise-corrupted sign function fluctuates between $\pm 1$. This results in some optimality gap as the algorithm decays toward the equilibrium. For points far from the equilibrium the evolution of the noise-corrupted dynamics is the same as the evolution of the noise-free dynamics. This justifies the use of the sign function to design resilient and robust-to-noise algorithms.

We make the following assumption on the links' weights:
\begin{ass} \label{ass_W}
The link weights are symmetric and positive, i.e., $W_{ij}=W_{ji}\geq0$ for $\forall i,j \in \{1,\dots,n\}$.
\end{ass}

The node/agent $i$ uses its local gradient information in its neighbourhood $\mc{N}_i$ and shares its objective gradient $\partial f_i(x_i)$ with the neighbours $j \in \mc{N}_i$. This justifies that the proposed dynamics \eqref{eq_sol} is a distributed and localized algorithm. Note that the dynamics \eqref{eq_sol} only constrain the weight matrix $W$ to be symmetric, and in contrast to \cite{rikos2021optimal,falsone2020tracking,falsone2017dual}, the weight matrix $W$ is not needed to be stochastic.

The convergence rate of the sign-based dynamics \eqref{eq_sol} is slow. To improve its convergence rate we modify the solution by some additive terms motivated by the fixed-time optimization protocols \cite{garg2020fixed,dai2022consensus}. The new \textit{accelerated} dynamics is as follows:
\begin{align} \nonumber
	\dot{x}_i = -\frac{1}{a_i}\sum_{j \in \mc{N}_i} &W_{ij}\Bigl( sgn\Big(\frac{\partial f_i(x_i)}{a_i} -  \frac{\partial f_j(x_j)}{a_j} + \nu_{ij} \Big)\\ &+  sgn^{\mu_1}\Big(\frac{\partial f_i(x_i)}{a_i} -  \frac{\partial f_j(x_j)}{a_j} + \nu_{ij} \Big) \nonumber \\ &+  sgn^{\mu_2}\Big(\frac{\partial f_i(x_i)}{a_i} -  \frac{\partial f_j(x_j)}{a_j} + \nu_{ij} \Big)\Bigr),
	\label{eq_sol_fixed}
\end{align}
with $0<\mu_1<1$ and  $\mu_2>1$. The function
$\mbox{sgn}^\mu(x)$ is defined as:
\begin{eqnarray}
	sgn^\mu(x)=x|x|^{\mu-1},
\end{eqnarray}

Next, we review some properties of the proposed solution. The first property is anytime-feasibility saying that the resource-demand balance holds at all evolution times of the dynamics. This holds for both dynamics \eqref{eq_sol} and \eqref{eq_sol_fixed}.

\begin{lem} \label{lem_feasible_intime}
    Assume that the nodes' states are initialized with feasible values, i.e., $\mb{a}^\top \mb{x}(0) = b$. In the absence of noise, the nodes' state values remain feasible throughout the evolution of system dynamics \eqref{eq_sol} and \eqref{eq_sol_fixed}, i.e., $\mb{a}^\top \mb{x}(t) = b$ for $t>0$.	
\end{lem}
\begin{proof}
	We prove the all-time feasibility for dynamics \eqref{eq_sol_fixed} and for the other one similarly follows. For the state dynamics \eqref{eq_sol_fixed} we have,
	\begin{align} \nonumber
		\frac{d}{dt}(\mb{a}^\top \mb{x}(t))=\sum_{i=1}^n a_i \dot{x}_i(t)& \\ \nonumber
		= -\sum_{i=1}^n \sum_{j \in \mc{N}_i} &W_{ij}\Bigl( sgn\Big(\frac{\partial f_i(x_i)}{a_i} -  \frac{\partial f_j(x_j)}{a_j}  \Big)\\ &+  sgn^{\mu_1}\Big(\frac{\partial f_i(x_i)}{a_i} -  \frac{\partial f_j(x_j)}{a_j}  \Big) \nonumber \\ &+  sgn^{\mu_2}\Big(\frac{\partial f_i(x_i)}{a_i} -  \frac{\partial f_j(x_j)}{a_j}  \Big)\Bigr). \label{eq_proof_feas}
	\end{align}
	Recall from Assumption~\ref{ass_W} that $W_{ij}=W_{ji}$ and $sgn^\mu(-x)=-sgn^\mu(x)$. Therefore, the summation in \eqref{eq_proof_feas} is equal to zero and  $\frac{d}{dt}(\mb{a}^\top \mb{x}(t))=0$. This implies that under the solution dynamics, the balance is preserved and is time-invariant. 		
\end{proof}

\begin{lem} \label{lem_tree}
	Let Assumptions~\ref{ass_W} hold. For any feasible initialization define the invariant set
	\begin{align}
		\mc{I} := \{\mb{x}| \mb{a}^\top \mb{x} = b , \nabla F(\mb{x}) \in \mbox{span}(\mb{a}) \}.
	\end{align}
    Then, in the absence of noise, the equilibrium under dynamics \eqref{eq_sol} and \eqref{eq_sol_fixed} is in $\mc{I}$.
\end{lem}
\begin{proof}
	The proof follows from symmetry in the weight matrix and the fact that $sgn^\mu(x)$ is an odd sign-preserving mapping. This implies that the conditions in \cite[Theorem~1]{ccta22} hold and the proof of the lemma follows.
\end{proof}

\begin{lem} \label{lem_sum}
	\cite{garg2019fixed2}
	Consider a nonlinear mapping $g(\mb{x}) $ and symmetric matrix $W$. Then, for any $\bm{\psi} \in \mathbb{R}$ the following holds,
	\begin{eqnarray} \nonumber
		\sum_{i =1}^n \bm{\psi}_i \sum_{j =1}^nW_{ij}g(\bm{\psi}_j-\bm{\psi}_i)= \sum_{i,j =1}^n \frac{W_{ij}}{2} (\bm{\psi}_j-\bm{\psi}_i)g(\bm{\psi}_j-\bm{\psi}_i).
	\end{eqnarray}
\end{lem}

We first prove convergence and stability in the absence of noise and then extend the results to the noise-corrupted case.

\begin{theorem} \label{thm_converg}
	Under Assumption~\ref{ass_W} and feasible initialization, the robust dynamics \eqref{eq_sol} and \eqref{eq_sol_fixed} converges to the invariant set $\mc{I}$ in Lemma~\ref{lem_tree} and converges toward the optimizer of problem~\ref{eq_dra}.
\end{theorem}
\begin{proof}
	Consider the residual function  $\overline{F}(\mb{x})=F(\mb{x})-F^*>0$ as  the Lyapunov function (with optimal value $F^*:=F(\mb{x}^*)$). Note that this function satisfies the conditions for a valid Lyapunov function, i.e., $\dot{\overline{F}}(\mb{x}^*) = 0$ and $\mb{x}^* \in \mc{I}$ is the equilibrium of $\overline{F}(\mb{x})$. To prove convergence we need to show that $\dot{\overline{F}}(\mb{x})$ decreases under dynamics \eqref{eq_sol} and \eqref{eq_sol_fixed}.
	\begin{align}
		\dot{\overline{F}}(\mb{x}) &=  \sum_{i =1}^n \partial f_i(x_i) \dot{x}_i\nonumber \\
		&= \sum_{i =1}^n -\frac{\partial f_i(x_i)}{a_i}\sum_{j \in \mc{N}_i} W_{ij} g\Big(\frac{\partial f_i(x_i)}{a_i} -  \frac{\partial f_j(x_j)}{a_j}\Big)\Bigr). \nonumber
	\end{align}
	We prove for dynamics \eqref{eq_sol_fixed} and the proof for dynamics \eqref{eq_sol} similarly follows. From Lemma~\ref{lem_sum},
	\begin{align} \nonumber
	\dot{\overline{F}}(\mb{x}) = -\sum_{i,j =1}^n W_{ij}&\Big(\frac{\partial f_i(\mb{x}_i)}{a_i} -  \frac{\partial f_j(\mb{x}_j)}{a_j}\Big) \\ \nonumber &\Bigl( sgn\Big(\frac{\partial f_i(x_i)}{a_i} -  \frac{\partial f_j(x_j)}{a_j}  \Big)\\ &+  sgn^{\mu_1}\Big(\frac{\partial f_i(x_i)}{a_i} -  \frac{\partial f_j(x_j)}{a_j}  \Big) \nonumber \\ &+  sgn^{\mu_2}\Big(\frac{\partial f_i(x_i)}{a_i} -  \frac{\partial f_j(x_j)}{a_j}  \Big)\Bigr).
\end{align}
	Recall that $sgn^{\mu}(x)$ is an odd sign-preserving function, we have $x sgn^{\mu}(x)\geq0$. This implies that
	 $\dot{\overline{F}}(\mb{x}) \leq0$ with strict equality at the invariant set $\mc{I}$. From LaSalle's invariance principle \cite{nonlin}, the proof follows.
\end{proof}

In the presence of noise $\nu_{ij}$ the solution converges to the proximity of the optimizer $\mb{x}^*$. $\varepsilon$-neighborhood of the optimizer is also the outcome of the chattering phenomena due to the non-Lipschitz nature of the proposed dynamics. From the discretized version of the dynamics~\eqref{eq_sol} we have,
\begin{align} \nonumber
	| x_i (t&+dt) - x_i(t) | \\ \nonumber
	&=\left\lvert -\frac{dt}{a_i}\sum_{j \in \mc{N}_i} W_{ij} sgn\Big(\frac{\partial f_i(x_i)}{a_i} -  \frac{\partial f_j(x_j)}{a_j} + \nu_{ij} \Big) \right\rvert \\ \label{eq_epsil}
	&\leq\frac{dt}{a_i} \sum_{j \in \mc{N}_i} W_{ij}
\end{align}
with $dt$ as the discretization time-step. In the chattering region, we have $|x_i (t+dt) |=|x_i(t)|=\varepsilon$, and one can compute the $\varepsilon$-neighborhood  bound as,
\begin{align} \label{eq_termination}
\varepsilon \leq \frac{dt}{2\min(a_i)} \sum_{j \in \mc{N}_i} W_{ij}.
\end{align}
Therefore, the $
\varepsilon$-proximity bound in the neighborhood of the optimizer  $\mb{x}^*$ can be defined. This gives the optimality gap of the dynamics~\eqref{eq_sol} (and similarly for dynamics~\eqref{eq_sol_fixed}).
Note that larger discrete step size $dt$ results in larger chattering and optimality gap. In other words, smaller $dt$ implies slower convergence with a smaller optimality gap around $\mb{x}^*$. As a remedy, to reduce the chattering effect, some works suggest adopting \textit{saturation}-based dynamics, for example, see \cite{ccta22}. In these works the solution converges in \textit{finite-time} to a proximity or $\varepsilon$-neighborhood of $\mb{x}^*$, while inside this proximity the solution slowly and asymptotically converges towards the optimizer.

The Algorithm~\ref{alg_ac} summarizes the proposed resource allocation solution.
Note that the termination criteria is defined, for example, as the time at which the average of states (or the sum of states) starts changing and is not constant anymore. This implies that the solution starts to violate the feasibility condition and, thus, the algorithm needs to stop iterating.
Before reaching the termination criteria the noise has no effect on the convergence toward the optimal point. This is a consequence of using sign function that removes the effect of noise at the regions not close to the equilibrium. After reaching this criteria the algorithm terminates and the computed resources are assigned to the agents. \textit{This follows the all-time feasibility of the solution before termination-time, implying that at the termination time the resource-demand balance holds.} Inside this region noise affects the convergence and optimality. This implies that in the presence of noise, the optimality gap is $\varepsilon$ and we cannot get closer to the $\mb{x}^*$. Note that, to compute the sum or average of local costs, there exists \textit{dynamic average consensus} algorithms that finds the sum/average of some time-varying reference signals in a decentralized way by sharing only local data, see for example \cite{kia2019tutorial,spanos2005dynamic,nosrati2012dynamic}. Typically, these algorithms are based on double time-scale scenarios with some inner loops for extra communications to deal with dynamic nature of the parameters.

\begin{algorithm}
	\KwData{Input:  $W$, $\mc{N}_i$, $f_i(\cdot)$}
	\KwResult{ Optimal point and optimal cost}
	{\textbf{Initialization:} $t=0$ and random feasible initialization}
	%Set the termination criteria as $\varepsilon$ defined in \eqref{eq_termination} \;
	\While{termination criteria NOT hold\;
	}{Node $i$ receives $\partial f_j (x_j(t))$ and $f_j (x_j(t))$ from neighbor nodes $j\in \mc{N}_i$\;
		Node $i$ updates $x_i(t)$ via dynamics~\eqref{eq_sol}\;
		Node $i$ shares $\partial f_i (x_i(t))$ and $f_i (x_i(t))$ with neighbour nodes $j \in \mc{N}_i$\;
	}
	\caption{\textsf{The Robust Resource Allocation Algorithm}} \label{alg_ac}
\end{algorithm}
\section{Simulations} \label{sec_sim}
In this section, we verify our theoretical results with some academic simulations.
We first consider the noise-free case to compare the convergence with some existing literature.  Consider a group of $n=40$ nodes/agents to be optimally assigned resources under the following non-quadratic cost function:
\begin{align}
	f_i^\sigma (\mb{x}_i) = f_i(\mb{x}_i) + h (\mb{x}_i - x_{max}) + h(x_{min} - \mb{x}_i)
\end{align}
with $h(u) := \max{(u,0)}$, $x_{max}=100$, $x_{min}=20$ as the box constraints and $f_i(x_i) = g_ix_i^2+c_ix_i$.
The parameters of the cost function are chosen randomly in the range: $g_i \in (0,0.05]$, $c_i \in (0,5]$.
The multi-agent network is considered a random network of Erdos-Renyi type with $20\%$ linking probability. The demand value is set equal to $b=2500$ and the weighting factor is set equal to $a_i = 1$. The evolution of residual error $\overline{F}$ (the Lyapunov function) for different existing resource allocation dynamics are compared in Fig.~\ref{fig_compare}. To give more details, linear dynamics \cite{boyd2006optimal}, accelerated linear solution \cite{shames2011accelerated}, finite-time convergent solution \cite{wang2020distributed}, and saturation-based dynamics \cite{ccta22} are given for comparison. We set the parameters $\mu_1=2$ and $\mu_2 = 0.4$ for accelerated dynamics~\eqref{eq_sol_fixed}.
\begin{figure}[]
	\centering
	\includegraphics[width=2.6in]{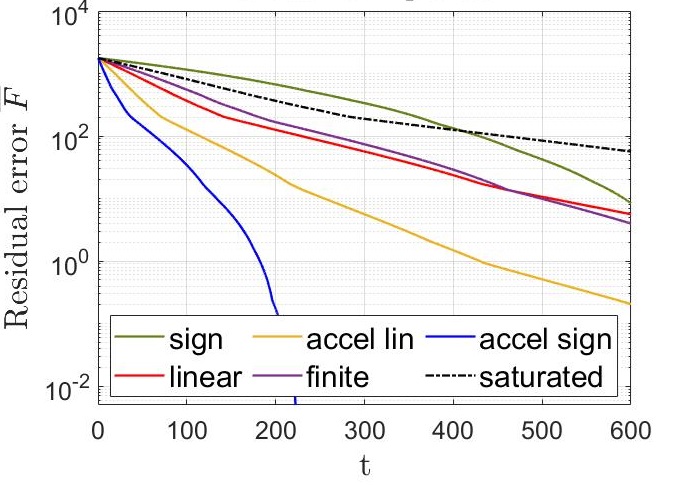}
	\caption{This figure shows a comparative simulation of different existing solutions with our proposed sign-based and accelerated sign-based dynamics. }
	\label{fig_compare}
\end{figure}

Next, we move to noise-corrupted dynamics to check the performance in the presence of random noise. We consider Gaussian distribution random noise $\nu_{ij} = \mc{N}(0,0.001)$ corrupting the dynamics \eqref{eq_sol}. We consider a similar setup as in the previous simulation. The time evolution of agents' states and the residual error is shown in Fig.~\ref{fig_noise}. The algorithm terminates with some optimality gap at time $329$. Before this time (the termination criteria) the noise has minor effects at the convergence and optimality, and the states move toward the optimizer and the residual is decaying. However, after this termination criteria, the states' evolution gets corrupted by noise and the state values fluctuate as can be seen from the figure. Before this point the average of states (representing the constraint feasibility) is constant. This implies all-time feasibility. After this termination point, the constraint feasibility does not hold anymore as the average of states changes due to the effect of noise. This implies that after the termination criteria, the algorithm loses its all-time feasibility.

\begin{figure}[]
	\centering
	\includegraphics[width=2.6in]{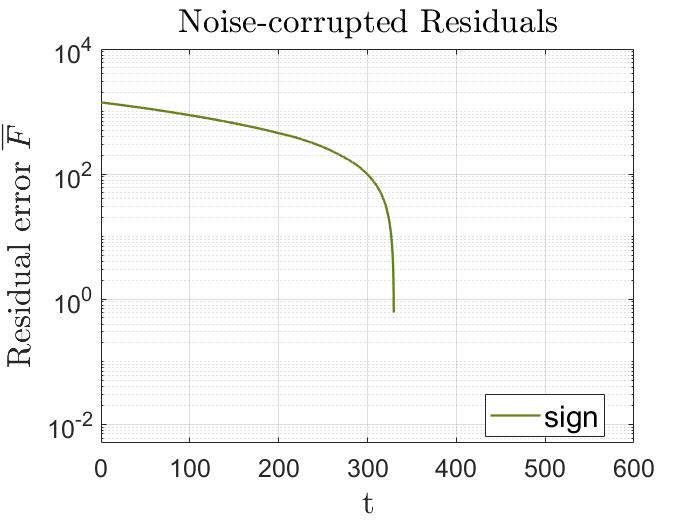}
	\includegraphics[width=2.6in]{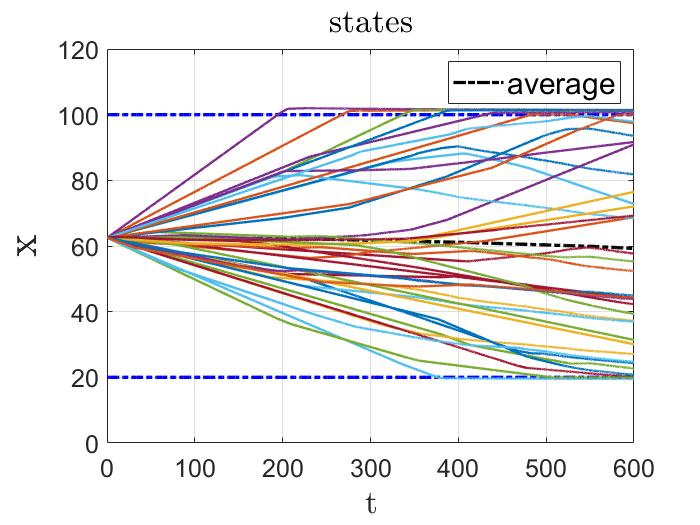}
	\caption{The evolution of the residual and all states under sign-based dynamics~\eqref{eq_sol} under noisy setups: The algorithm terminates before the exact convergence and converges with some optimality gap to the proximity of the optimizer (note that resource-demand feasibility holds up-to termination time). }
	\label{fig_noise}
\end{figure}

\section{Conclusion and Future Works} \label{sec_con}
%\subsection{Concluding Remarks}
This work proposes distributed (or networked) sign-based dynamics to solve resource allocation in noisy setups with a certain optimality gap. We further improve the convergence rate of the solution by proposing accelerated dynamics containing additive nonlinear terms. The proposed solution converges over general uniformly-connected networks over which the agents may even lose connection with other agents at some points while regaining the network connectivity over finite time intervals. The solution preserves all-time feasibility in regions far from the equilibrium, which is another privilege over the existing dual-based formulations, e.g., scenarios based on alternating-direction-method-of-multipliers (ADMM) \cite{jiang2022distributed,cdc22,jian2019distributed,farakhor2023distributed}.

%\subsection{Future Directions}
Extending the results to consider packet drops and link removal over unreliable networks \cite{icrom2022} is one direction of future research. Convergence in the presence of time delay over the multi-agent network \cite{OJCSYS,scl23} is another promising research direction. Addressing privacy \cite{rikos2022distributed} and security \cite{wu2023distributed} concerns is also an interesting research direction. Other applications such as network utility maximization problem \cite{nekouei2016convergence,tsai2022network} are also of interest.

\bibliographystyle{IEEEbib}
\bibliography{bibliography}

\begin{thebibliography}{10}

\bibitem{acc13}
M.~Doostmohammadian and U.~A. Khan,
\newblock ``Topology design in networked estimation: A generic approach,''
\newblock in {\em American Control Conference}. IEEE, 2013, pp. 4134--4139.

\bibitem{das2015distributed}
S.~Das and J.~M.~F. Moura,
\newblock ``Distributed kalman filtering with dynamic observations consensus,''
\newblock {\em IEEE Transactions on Signal Processing}, vol. 63, no. 17, pp.
  4458--4473, 2015.

\bibitem{deghat2019detection}
M.~Deghat, V.~Ugrinovskii, I.~Shames, and C.~Langbort,
\newblock ``Detection and mitigation of biasing attacks on distributed
  estimation networks,''
\newblock {\em Automatica}, vol. 99, pp. 369--381, 2019.

\bibitem{icassp13}
M.~Doostmohammadian and U.~A. Khan,
\newblock ``On the distributed estimation of rank-deficient dynamical systems:
  A generic approach,''
\newblock in {\em IEEE International Conference on Acoustics, Speech and Signal
  Processing}, 2013, pp. 4618--4622.

\bibitem{ennasr2016distributed}
O.~Ennasr, G.~Xing, and X.~Tan,
\newblock ``Distributed time-difference-of-arrival {(TDOA)}-based localization
  of a moving target,''
\newblock in {\em IEEE 55th Conference on Decision and Control}. IEEE, 2016,
  pp. 2652--2658.

\bibitem{tase}
M.~Doostmohammadian, A.~Taghieh, and H.~Zarrabi,
\newblock ``Distributed estimation approach for tracking a mobile target via
  formation of {UAVs},''
\newblock {\em IEEE Transactions on Automation Science and Engineering}, vol.
  19, no. 4, pp. 3765--3776, 2021.

\bibitem{cortes2004coverage}
J.~Cortes, S.~Martinez, T.~Karatas, and F.~Bullo,
\newblock ``Coverage control for mobile sensing networks,''
\newblock {\em IEEE Transactions on robotics and Automation}, vol. 20, no. 2,
  pp. 243--255, 2004.

\bibitem{msc09}
M.~Doostmohammadian, H.~Sayyaadi, and M.~Moarref,
\newblock ``A novel consensus protocol using facility location algorithms,''
\newblock in {\em IEEE Control Applications \& Intelligent Control Conference},
  2009, pp. 914--919.

\bibitem{sayyaadi2010distributed}
H.~Sayyaadi and M.~Moarref,
\newblock ``A distributed algorithm for proportional task allocation in
  networks of mobile agents,''
\newblock {\em IEEE Transactions on Automatic Control}, vol. 56, no. 2, pp.
  405--410, 2010.

\bibitem{csl2021}
Mohammadreza Doostmohammadian, Alireza Aghasi, Themistoklis Charalambous, and
  Usman~A Khan,
\newblock ``Distributed support vector machines over dynamic balanced directed
  networks,''
\newblock {\em IEEE Control Systems Letters}, vol. 6, pp. 758--763, 2021.

\bibitem{xin2020decentralized}
R.~Xin, S.~Kar, and U.~A. Khan,
\newblock ``Decentralized stochastic optimization and machine learning: A
  unified variance-reduction framework for robust performance and fast
  convergence,''
\newblock {\em IEEE Signal Processing Magazine}, vol. 37, no. 3, pp. 102--113,
  2020.

\bibitem{zayyani2016distributed}
H.~Zayyani, M.~Korki, and F.~Marvasti,
\newblock ``A distributed 1-bit compressed sensing algorithm robust to
  impulsive noise,''
\newblock {\em IEEE Communications Letters}, vol. 20, no. 6, pp. 1132--1135,
  2016.

\bibitem{stankovic2020nonlinear}
S.~S. Stankovi{\'c}, M.~Beko, and M.~S. Stankovi{\'c},
\newblock ``Nonlinear robustified stochastic consensus seeking,''
\newblock {\em Systems \& Control Letters}, vol. 139, pp. 104667, 2020.

\bibitem{taes}
M.~Doostmohammadian,
\newblock ``Single-bit consensus with finite-time convergence: Theory and
  applications,''
\newblock {\em IEEE Transactions on Aerospace and Electronic Systems}, vol. 56,
  no. 4, pp. 3332--3338, 2020.

\bibitem{stankovic2019robust}
S.~S. Stankovi{\'c}, M.~Beko, and M.~S. Stankovi{\'c},
\newblock ``Robust nonlinear consensus seeking,''
\newblock in {\em IEEE 58th Conference on Decision and Control}. IEEE, 2019,
  pp. 4465--4470.

\bibitem{sadigh2022proportionate}
A.~Naeimi~Sadigh and H.~Zayyani,
\newblock ``A proportionate robust diffusion recursive least exponential
  hyperbolic cosine algorithm for distributed estimation,''
\newblock {\em IEEE Transactions on Circuits and Systems II: Express Briefs},
  vol. 69, no. 4, pp. 2381--2385, 2022.

\bibitem{babarsad2019analytic}
S.~Bakhshandeh~Babarsad, S.~M. Saberali, and M.~Majidi,
\newblock ``Analytic performance investigation of signal level estimator based
  on empirical characteristic function in impulsive noise,''
\newblock {\em Digital Signal Processing}, vol. 92, pp. 20--25, 2019.

\bibitem{zandi2022diffusion}
S.~Zandi and M.~Korki,
\newblock ``Diffusion normalized maximum versoria criterion robust to impulsive
  noise,''
\newblock {\em IEEE Transactions on Circuits and Systems II: Express Briefs},
  vol. 70, no. 4, pp. 1660--1664, 2022.

\bibitem{baranwal2020robust}
M.~Baranwal, K.~Garg, D.~Panagou, and A.~O. Hero,
\newblock ``Robust distributed fixed-time economic dispatch under time-varying
  topology,''
\newblock {\em IEEE Control Systems Letters}, vol. 5, no. 4, pp. 1183--1188,
  2020.

\bibitem{kar2012distributed}
S.~Kar and G.~Hug,
\newblock ``Distributed robust economic dispatch in power systems: A consensus+
  innovations approach,''
\newblock in {\em IEEE Power and Energy Society General Meeting}, 2012, pp.
  1--8.

\bibitem{doan2020distributed}
T.~T. Doan and C.~L. Beck,
\newblock ``Distributed resource allocation over dynamic networks with
  uncertainty,''
\newblock {\em IEEE Transactions on Automatic Control}, vol. 66, no. 9, pp.
  4378--4384, 2020.

\bibitem{ferguson2023robust}
B.~L. Ferguson and J.~R. Marden,
\newblock ``Robust utility design in distributed resource allocation problems
  with defective agents,''
\newblock {\em Dynamic Games and Applications}, vol. 13, no. 1, pp. 208--230,
  2023.

\bibitem{xu2022robust}
D.~Xu, X.~Yu, D.~W.~K. Ng, A.~Schmeink, and R.~Schober,
\newblock ``Robust and secure resource allocation for isac systems: A novel
  optimization framework for variable-length snapshots,''
\newblock {\em IEEE Transactions on Communications}, vol. 70, no. 12, pp.
  8196--8214, 2022.

\bibitem{ding2021differentially}
T.~Ding, S.~Zhu, C.~Chen, J.~Xu, and X.~Guan,
\newblock ``Differentially private distributed resource allocation via
  deviation tracking,''
\newblock {\em IEEE Transactions on Signal and Information Processing over
  Networks}, vol. 7, pp. 222--235, 2021.

\bibitem{turan2020resilient}
B.~Turan, C.~A. Uribe, H.~Wai, and M.~Alizadeh,
\newblock ``Resilient primal--dual optimization algorithms for distributed
  resource allocation,''
\newblock {\em IEEE Transactions on Control of Network Systems}, vol. 8, no. 1,
  pp. 282--294, 2020.

\bibitem{ramirez2016gradient}
E.~Ram{\'\i}rez-Llanos and S.~Mart{\'\i}nez,
\newblock ``Gradient-free distributed resource allocation via simultaneous
  perturbation,''
\newblock in {\em 54th Allerton Conference on Communication, Control, and
  Computing}, 2016, pp. 590--595.

\bibitem{alaviani2021reciprocity}
S.~S. Alaviani, A.~G. Kelkar, and U.~Vaidya,
\newblock ``Reciprocity of algorithms solving distributed consensus-based
  optimization and distributed resource allocation,''
\newblock in {\em 29th Mediterranean Conference on Control and Automation}.
  IEEE, 2021, pp. 904--909.

\bibitem{ecc22}
M.~Doostmohammadian, A.~Aghasi, M.~Pirani, E.~Nekouei, U.~A. Khan, and
  T.~Charalambous,
\newblock ``Fast-convergent anytime-feasible dynamics for distributed
  allocation of resources over switching sparse networks with quantized
  communication links,''
\newblock in {\em European Control Conference}. IEEE, 2022, pp. 84--89.

\bibitem{ccta22}
M.~Doostmohammadian, A.~Aghasi, M.~Vrakopoulou, and T.~Charalambous,
\newblock ``1st-order dynamics on nonlinear agents for resource allocation over
  uniformly-connected networks,''
\newblock in {\em IEEE Conference on Control Technology and Applications}.
  IEEE, 2022, pp. 1184--1189.

\bibitem{cherukuri2015distributed}
A.~Cherukuri and J.~Cort{\'e}s,
\newblock ``Distributed generator coordination for initialization and anytime
  optimization in economic dispatch,''
\newblock {\em IEEE Transactions on Control of Network Systems}, vol. 2, no. 3,
  pp. 226--237, 2015.

\bibitem{cherukuri2017distributed}
A.~Cherukuri and J.~Cort{\'e}s,
\newblock ``Distributed coordination of {DERs} with storage for dynamic
  economic dispatch,''
\newblock {\em IEEE transactions on Automatic Control}, vol. 63, no. 3, pp.
  835--842, 2017.

\bibitem{LCSS23}
M.~Doostmohammadian,
\newblock ``Distributed energy resource management: All-time resource-demand
  feasibility, delay-tolerance, nonlinearity, and beyond,''
\newblock {\em IEEE Control Systems Letters}, pp. 3423--3428, 2023.

\bibitem{grammenos2023cpu}
A.~Grammenos, T.~Charalambous, and E.~Kalyvianaki,
\newblock ``{CPU} scheduling in data centers using asynchronous finite-time
  distributed coordination mechanisms,''
\newblock {\em IEEE Transactions on Network Science and Engineering}, 2023.

\bibitem{OJCSYS}
M.~Doostmohammadian, A.~Aghasi, A.~I. Rikos, A.~Grammenos, E.~Kalyvianaki,
  C.~N. Hadjicostis, K.~H. Johansson, and T.~Charalambous,
\newblock ``Distributed anytime-feasible resource allocation subject to
  heterogeneous time-varying delays,''
\newblock {\em IEEE Open Journal of Control Systems}, vol. 1, pp. 255--267,
  2022.

\bibitem{rikos2021optimal}
A.~I. Rikos, A.~Grammenos, E.~Kalyvianaki, C.~N. Hadjicostis, T.~Charalambous,
  and K.~H. Johansson,
\newblock ``Optimal {CPU} scheduling in data centers via a finite-time
  distributed quantized coordination mechanism,''
\newblock in {\em 60th IEEE Conference on Decision and Control}. IEEE, 2021,
  pp. 6276--6281.

\bibitem{falsone2020tracking}
A.~Falsone, I.~Notarnicola, G.~Notarstefano, and M.~Prandini,
\newblock ``Tracking-{ADMM} for distributed constraint-coupled optimization,''
\newblock {\em Automatica}, vol. 117, pp. 108962, 2020.

\bibitem{mukherjee2016distributed}
J.~Mukherjee and A.~Gupta,
\newblock ``Distributed charge scheduling of plug-in electric vehicles using
  inter-aggregator collaboration,''
\newblock {\em IEEE Transactions on Smart Grid}, vol. 8, no. 1, pp. 331--341,
  2016.

\bibitem{vujanic2016decomposition}
R.~Vujanic, P.~M. Esfahani, P.~J. Goulart, S.~Mari{\'e}thoz, and M.~Morari,
\newblock ``A decomposition method for large scale {MILPs}, with performance
  guarantees and a power system application,''
\newblock {\em Automatica}, vol. 67, pp. 144--156, 2016.

\bibitem{falsone2017dual}
A.~Falsone, K.~Margellos, S.~Garatti, and M.~Prandini,
\newblock ``Dual decomposition for multi-agent distributed optimization with
  coupling constraints,''
\newblock {\em Automatica}, vol. 84, pp. 149--158, 2017.

\bibitem{boyd2006optimal}
L.~Xiao and S.~Boyd,
\newblock ``Optimal scaling of a gradient method for distributed resource
  allocation,''
\newblock {\em Journal of Optimization Theory and Applications}, vol. 129, no.
  3, pp. 469--488, 2006.

\bibitem{shames2011accelerated}
E.~Ghadimi, M.~Johansson, and I.~Shames,
\newblock ``Accelerated gradient methods for networked optimization,''
\newblock in {\em IEEE American Control Conference}, 2011, pp. 1668--1673.

\bibitem{nesterov1998introductory}
Y.~Nesterov,
\newblock ``Introductory lectures on convex programming, {volume I}: Basic
  course,''
\newblock {\em Lecture notes}, vol. 3, no. 4, pp. 5, 1998.

\bibitem{bertsekas1975necessary}
D.~P. Bertsekas,
\newblock ``Necessary and sufficient conditions for a penalty method to be
  exact,''
\newblock {\em Mathematical programming}, vol. 9, no. 1, pp. 87--99, 1975.

\bibitem{smith1997penalty}
A.~E. Smith, D.~W. Coit, T.~Baeck, D.~Fogel, and Z.~Michalewicz,
\newblock ``Penalty functions,''
\newblock {\em Handbook of evolutionary computation}, vol. 97, no. 1, pp. C5,
  1997.

\bibitem{garg2020fixed}
K.~Garg and D.~Panagou,
\newblock ``Fixed-time stable gradient flows: Applications to continuous-time
  optimization,''
\newblock {\em IEEE Transactions on Automatic Control}, vol. 66, no. 5, pp.
  2002--2015, 2020.

\bibitem{dai2022consensus}
H.~Dai, X.~Fang, and J.~Jia,
\newblock ``Consensus-based distributed fixed-time optimization for a class of
  resource allocation problems,''
\newblock {\em Journal of the Franklin Institute}, vol. 359, no. 18, pp.
  11135--11154, 2022.

\bibitem{garg2019fixed2}
K.~Garg, M.~Baranwal, A.~O. Hero, and D.~Panagou,
\newblock ``Fixed-time distributed optimization under time-varying
  communication topology,''
\newblock {\em arXiv preprint arXiv:1905.10472}, 2019.

\bibitem{nonlin}
J.J. Slotine and W.~Li,
\newblock {\em Applied nonlinear control},
\newblock Prentice-Hall, 1991.

\bibitem{kia2019tutorial}
S.~S. Kia, B.~Van~Scoy, J.~Cortes, R.~A. Freeman, K.~M. Lynch, and S.~Martinez,
\newblock ``Tutorial on dynamic average consensus: The problem, its
  applications, and the algorithms,''
\newblock {\em IEEE Control Systems Magazine}, vol. 39, no. 3, pp. 40--72,
  2019.

\bibitem{spanos2005dynamic}
D.~P. Spanos, R.~Olfati-Saber, and R.~M. Murray,
\newblock ``Dynamic consensus on mobile networks,''
\newblock in {\em IFAC world congress}, 2005, pp. 1--6.

\bibitem{nosrati2012dynamic}
S.~Nosrati, M.~Shafiee, and M.~B. Menhaj,
\newblock ``Dynamic average consensus via nonlinear protocols,''
\newblock {\em Automatica}, vol. 48, no. 9, pp. 2262--2270, 2012.

\bibitem{wang2020distributed}
B.~Wang, Q.~Fei, and Q.~Wu,
\newblock ``Distributed time-varying resource allocation optimization based on
  finite-time consensus approach,''
\newblock {\em IEEE Control Systems Letters}, vol. 5, no. 2, pp. 599--604,
  2020.

\bibitem{jiang2022distributed}
W.~Jiang, M.~Doostmohammadian, and T.~Charalambous,
\newblock ``Distributed resource allocation via {ADMM} over digraphs,''
\newblock in {\em IEEE 61st Conference on Decision and Control (CDC)}. IEEE,
  2022, pp. 5645--5651.

\bibitem{cdc22}
M.~Doostmohammadian, W.~Jiang, and T.~Charalambous,
\newblock ``{DTAC-ADMM}: Delay-tolerant augmented consensus {ADMM}-based
  algorithm for distributed resource allocation,''
\newblock in {\em IEEE 61st Conference on Decision and Control}. IEEE, 2022,
  pp. 308--315.

\bibitem{jian2019distributed}
L.~Jian, J.~Hu, J.~Wang, and K.~Shi,
\newblock ``Distributed inexact dual consensus {ADMM} for network resource
  allocation,''
\newblock {\em Optimal Control Applications and Methods}, vol. 40, no. 6, pp.
  1071--1087, 2019.

\bibitem{farakhor2023distributed}
A.~Farakhor, Y.~Wang, D.~Wu, and H.~Fang,
\newblock ``Distributed optimal power management for battery energy storage
  systems: A novel accelerated tracking admm approach,''
\newblock in {\em American Control Conference (ACC)}. IEEE, 2023, pp.
  3106--3112.

\bibitem{icrom2022}
M.~Doostmohammadian, U.~A. Khan, and A.~Aghasi,
\newblock ``Distributed constraint-coupled optimization over unreliable
  networks,''
\newblock in {\em 10th RSI International Conference on Robotics and
  Mechatronics (ICRoM)}. 2022, pp. 371--376, IEEE.

\bibitem{scl23}
M.~Doostmohammadian, A.~Aghasi, M.~Vrakopoulou, H.~R. Rabiee, U.~A. Khan, and
  T.~Charalambous,
\newblock ``Distributed delay-tolerant strategies for equality-constraint
  sum-preserving resource allocation,''
\newblock {\em Systems \& Control Letters}, vol. 182, pp. 105657, 2023.

\bibitem{rikos2022distributed}
A.~I. Rikos, T.~Charalambous, K.~H. Johansson, and C.~N. Hadjicostis,
\newblock ``Distributed event-triggered algorithms for finite-time
  privacy-preserving quantized average consensus,''
\newblock {\em IEEE Transactions on Control of Network Systems}, vol. 10, no.
  1, pp. 38--50, 2022.

\bibitem{wu2023distributed}
X.~Wu, S.~Magn{\'u}sson, and M.~Johansson,
\newblock ``Distributed safe resource allocation using barrier functions,''
\newblock {\em Automatica}, vol. 153, pp. 111051, 2023.

\bibitem{nekouei2016convergence}
E.~Nekouei, T.~Alpcan, G.~N. Nair, and R.~J. Evans,
\newblock ``Convergence analysis of quantized primal-dual algorithms in network
  utility maximization problems,''
\newblock {\em IEEE Transactions on Control of Network Systems}, vol. 5, no. 1,
  pp. 284--297, 2016.

\bibitem{tsai2022network}
C.~Tsai,
\newblock {\em Network Utility Maximization Based on Information Freshness},
\newblock Ph.D. thesis, Purdue University, 2022.

\end{thebibliography}

\end{document}